\documentclass[a4paper,11pt]{article}

\usepackage{fullpage}
\usepackage{times}
\usepackage{soul}
\usepackage{url}
\usepackage[utf8]{inputenc}
\usepackage[small]{caption}
\usepackage{graphicx}
\usepackage{xcolor}
\usepackage{amsmath}
\usepackage{booktabs}
\usepackage{algorithm}
\usepackage{enumitem}

\usepackage{amsthm}
\usepackage{amssymb}
\usepackage{algpseudocode}
\usepackage{cleveref}

\usepackage{bbm}

\newtheorem{theorem}{Theorem}

\newtheorem{lemma}[theorem]{Lemma}

\theoremstyle{definition}

\usepackage{natbib}

\allowdisplaybreaks

\usepackage{authblk}

\title{\bf Extending the Characterization of Maximum Nash Welfare}

\author{Sheung Man Yuen and Warut Suksompong}

\affil{National University of Singapore}

\date{\vspace{-10mm}}

\begin{document}

\maketitle

\begin{abstract}
In the allocation of indivisible goods, the maximum Nash welfare rule has recently been characterized as the only rule within the class of additive welfarist rules that satisfies envy-freeness up to one good.
We extend this characterization to the class of all welfarist rules.
\end{abstract}

\section{Introduction} \label{sec:intro}

The fair allocation of indivisible goods---be it artwork, furniture, school supplies, or electronic devices---is a ubiquitous problem in society and has attracted significant interest in economics 
\citep{Moulin19}.
Among the plethora of methods that one may use to allocate indivisible goods fairly, the method that has arguably received the most attention in recent years is the \emph{maximum Nash welfare (MNW)} rule.
For instance, MNW is used to allocate goods on the popular fair division website Spliddit,\footnote{http://www.spliddit.org} which has served hundreds of thousands of users since its launch in 2014.

MNW selects from each profile an allocation that maximizes the product of the agents' utilities, or equivalently, the sum of their logarithms.
In an influential work, \citet{CaragiannisKuMo19} showed that every allocation output by MNW satisfies \emph{envy-freeness up to one good (EF1)}: given any two agents, if the first agent envies the second agent, then this envy can be eliminated by removing some good in the second agent's bundle.
Recently, \citet{Suksompong23} provided the first characterization of MNW by showing that it is the unique \emph{additive welfarist rule} that guarantees EF1---an additive welfarist rule selects an allocation maximizing a welfare notion that can be expressed as the sum of some function of the agents' utilities.
Suksompong's characterization raises an obvious question: Is MNW also the unique (not necessarily additive) \emph{welfarist rule} that guarantees EF1, where a welfarist rule selects an allocation maximizing a welfare notion that can be expressed as some function of the agents' utilities?

In this note, we answer the above question in the affirmative, by extending the characterization of \citet{Suksompong23} to the class of \emph{all} welfarist rules (whether additive or not).
This further solidifies the ``unreasonable fairness'' of MNW established by \citet{CaragiannisKuMo19}.

\section{Preliminaries} \label{sec:prelim}

Let $N = \{1, \ldots, n\}$ be the set of agents, and $G = \{g_1, \ldots, g_m\}$ be the set of goods.
Each agent $i \in N$ has a \emph{utility function} $u_i : 2^G \to \mathbb{R}_{\geq 0}$; for simplicity, we write $u_i(g)$ instead of $u_i (\{g\})$ for a single good $g \in G$. 
We assume that the utility functions are additive, that is, $u_i (G') = \sum_{g \in G'} u_i (g)$ for all $i \in N$ and $G' \subseteq G$. 
A \emph{profile} consists of $N$, $G$ and $(u_i)_{i \in N}$. 

An \emph{allocation} $A = (A_1, \ldots, A_n)$ is an ordered partition of $G$ into $n$ bundles such that bundle $A_i$ is allocated to agent $i$.
An agent $i\in N$ receives utility $u_i(A_i)$ from allocation~$A$.
An allocation~$A$ is \emph{EF1} if for every pair $i, j \in N$ such that $A_j \neq \emptyset$, there exists a good $g \in A_j$ with the property that $u_i (A_i) \geq u_i (A_j \setminus \{g\})$. 
A \emph{rule} maps any given profile to an allocation. 

Given $n \geq 2$, a \emph{welfare function} is a non-decreasing function $f_n : [0, \infty)^n \to [-\infty, \infty)$.
The \emph{welfarist rule with (welfare) function $f_n$} chooses from each profile an allocation $A$ that maximizes the \emph{welfare} $f_n(u_1 (A_1), \ldots, u_n (A_n))$; if there are multiple such allocations, the rule may choose one arbitrarily.

\section{The Result} \label{sec:result}

Before proceeding to our characterization, we first establish a technical lemma.

\begin{lemma} \label{lem:function}
Fix $n \geq 2$. Let $f_n : [0, \infty)^n \to [-\infty, \infty)$ be a function that is continuous on $(0, \infty)^n$.
Suppose that 
\begin{align} \label{eq:k_equal}
f_n((k+1)x_1, x_2, \ldots, x_{i-1}, kx_i, x_{i+1}, \ldots, x_n) = f_n(kx_1, x_2, \ldots, x_{i-1}, (k+1)x_i, x_{i+1}, \ldots, x_n)
\end{align}
for all $x_1, \ldots, x_n > 0$, positive integers $k$, and $i \in N \setminus \{1\}$. 
Then, there exists a continuous function $q: (0, \infty) \to [-\infty, \infty)$ such that $f_n(x_1, x_2, \ldots, x_n) = q(x_1 x_2 \cdots x_n)$ for all $x_1, \ldots, x_n > 0$.
\end{lemma}

\begin{proof}
Suppose that $f_n$ fulfills assumption \eqref{eq:k_equal}.
First, we show that $f_n$ satisfies
\begin{align} \label{eq:constant}
    f_n(x, x_2, \ldots, x_{i-1}, z/x, x_{i+1}, \ldots, x_n) = f_n(y, x_2, \ldots, x_{i-1}, z/y, x_{i+1}, \ldots, x_n)
\end{align}
for all $i \in N \setminus \{1\}$ and $x_2, \ldots, x_{i-1}, x_{i+1}, \ldots, x_n, x, y, z > 0$. 
Assume without loss of generality that $i = 2$; the proof for any other $i\in\{3,\ldots,n\}$ is analogous. 
Let $x_3, \ldots, x_n, z > 0$ be fixed throughout. 
Define $h : (0, \infty) \to [-\infty, \infty)$ by $h(x) := f_n(x, z/x, x_3, \ldots, x_n)$ for all $x > 0$. 
Note that $h$ is continuous due to the continuity of $f_n$ on $(0, \infty)^n$. 
For any positive integer $k$ and any $x > 0$, we have
\begin{align*}
    h \left( \frac{k+1}{k} \cdot x \right) &= f_n \left( (k+1) \cdot \frac{x}{k},\, k \cdot \frac{z}{(k+1)x},\, x_3, \ldots, x_n \right) \\
    &= f_n \left( k \cdot \frac{x}{k},\, (k+1) \cdot \frac{z}{(k+1)x},\, x_3, \ldots, x_n \right) \tag*{(by \eqref{eq:k_equal})} \\
    &= f_n (x, z/x, x_3, \ldots, x_n) \\
    &= h(x),
\end{align*}
so for any rational number $r = a/b > 1$, we have
\begin{align*}
    h(rx) 
    = h\left( \frac{a}{a-1} \cdot \frac{a-1}{a-2} \cdot \; \cdots \; \cdot \frac{b+1}{b} \cdot x \right) 
    = h\left( \frac{a-1}{a-2} \cdot \; \cdots \; \cdot \frac{b+1}{b} \cdot x \right) 
    = \cdots 
    = h(x).
\end{align*}
Similarly, we have $h(rx) = h(x)$ for any rational number $0 < r < 1$, hence the same equation is true for all positive rational numbers $r$. Since $h$ is continuous and the positive rational numbers are dense in $(0, \infty)$, we can conclude that $h$ is constant, and thus, $f_n(x, z/x, x_3, \ldots, x_n) = f_n(y, z/y, x_3, \ldots, x_n)$ for all $x, y > 0$. Hence, \eqref{eq:constant} is true for all $i \in N \setminus \{1\}$ and $x_2, \ldots, x_{i-1}, x_{i+1}, \ldots, x_n, x, y, z > 0$.

Next, we prove by backward induction that for all integers $1 \leq k \leq n$, there exists a continuous function $q_k : (0, \infty)^k \to [-\infty, \infty)$ such that $f_n(x_1, \ldots, x_n) = q_k(x_1 x_{k+1} \cdots x_n, x_2, \ldots, x_k)$ for all $x_1, \ldots, x_n > 0$. Then, $q := q_1$ gives the desired conclusion.

For the base case $k = n$, we have $q_n := f_n|_{(0, \infty)^n}$. 
For the inductive step, let $2 \leq k \leq n$ be given, and assume that such a function $q_k$ exists; we shall prove that $q_{k-1}$ exists as well. 
Define $q_{k-1}$ by $q_{k-1}(y_1, \ldots, y_{k-1}) := q_k(y_1, \ldots, y_{k-1}, 1)$ for all $y_1, \ldots, y_{k-1} > 0$. Note that $q_{k-1}$ is continuous on $(0, \infty)^{k-1}$ due to the continuity of $q_k$ on $(0, \infty)^k$. Let $x_1, \ldots, x_n > 0$ be given. Then, by setting $x := x_1$ and $y := z := x_1 x_k$, we have
\begin{align*}
    f_n(x_1, \ldots, x_n) &= f_n(x, x_2, \ldots, x_{k-1}, z/x, x_{k+1}, \ldots, x_n) \\
    &= f_n(y, x_2, \ldots, x_{k-1}, z/y, x_{k+1}, \ldots, x_n) \tag*{(by \eqref{eq:constant})} \\
    &= f_n(x_1 x_k, x_2, \ldots, x_{k-1}, 1, x_{k+1}, \ldots, x_n) \\
    &= q_k(x_1 x_k x_{k+1} \cdots x_n, x_2, \ldots, x_{k-1}, 1) \tag*{(by the inductive hypothesis)} \\
    &= q_{k-1}(x_1 x_k \cdots x_n, x_2, \ldots, x_{k-1}),
\end{align*}
establishing the inductive step and therefore the lemma.
\end{proof}

We now state our characterization.
Recall from \Cref{sec:prelim} that a welfare function is assumed to be non-decreasing on $[0,\infty)^n$.

\begin{theorem} 
\label{thm:main}
Fix $n \geq 2$. Let $f_n$ be a welfare function that is continuous\footnote{In the prior characterization of \emph{additive} welfarist rules, \citet{Suksompong23} made the stronger assumption that the welfare function is differentiable. Here, we only assume that the function is continuous.} and strictly increasing\footnote{The theorem does not hold without the assumption that $f_n$ is strictly increasing on $(0,\infty)^n$: for example, if $f_n$ is a constant function, then statement (b) holds but (a) does not.} on $(0, \infty)^n$. Then, the following three statements are equivalent:
\begin{enumerate}[label=(\alph*)]
    \item For every profile that admits an allocation where every agent receives positive utility, every allocation that can be chosen by the welfarist rule with function $f_n$ is EF1.
    \item For every profile that admits an allocation where every agent receives positive utility, there exists an EF1 allocation that can be chosen by the welfarist rule with function $f_n$.
    \item The following two statements hold for $f_n$:
    \begin{enumerate}[label=(\roman*)]
        \item There exists a strictly increasing and continuous function $q : (0, \infty) \to (-\infty, \infty)$ such that $f_n(x_1, x_2, \ldots, x_n) = q(x_1 x_2 \cdots x_n)$ for all $x_1, \ldots, x_n > 0$.
        \item The inequality $f_n(x_1, x_2, \ldots, x_n) > f_n(y_1, y_2, \ldots, y_n)$ holds for all $x_1, \ldots, x_n > 0$ and $y_1, \ldots, y_n \geq 0$ satisfying $\prod_{i=1}^n y_i = 0$.
    \end{enumerate}
\end{enumerate}
\end{theorem}

Note that if $f_n$ satisfies (c), then given any profile that admits an allocation where every agent receives positive utility, an allocation can be chosen by the welfarist rule with function $f_n$ if and only if it can be chosen by MNW, so the two rules are effectively equivalent.\footnote{For profiles that do not admit an allocation where every agent receives positive utility, MNW requires an additional tie-breaking specification in order to ensure EF1 \citep{CaragiannisKuMo19}.}
Hence, \Cref{thm:main} provides a characterization of MNW among all welfarist rules.

\begin{proof}[Proof of \Cref{thm:main}]
The implication (a) $\Rightarrow$ (b) is trivial. 
For the implication (c) $\Rightarrow$ (a), if $f_n$ satisfies (c), then given a profile that admits an allocation where every agent receives positive utility, every allocation that can be chosen by the welfarist rule with function $f_n$ is also an allocation that can be chosen by MNW, which is known to be EF1 \citep{CaragiannisKuMo19}; hence, $f_n$ also satisfies (a).
It therefore remains to prove the implication (b)~$\Rightarrow$~(c).
Assume that $f_n$ satisfies (b); we will show that both statements (i) and~(ii) of~(c) hold.

To prove (i), it suffices to show that $f_n$ satisfies \eqref{eq:k_equal} for all $x_1, \ldots, x_n > 0$, positive integers $k$, and $i \in N \setminus \{1\}$. 
Indeed, once this is shown, \Cref{lem:function} provides a continuous function $q: (0, \infty) \to [-\infty, \infty)$ satisfying $f_n(x_1, x_2, \ldots, x_n) = q(x_1 x_2 \cdots x_n)$ for all $x_1, \ldots, x_n > 0$.
Note that $q$ must be strictly increasing because $f_n$ is strictly increasing on $(0, \infty)^n$, and $-\infty$ cannot be in the range of $q$ since $q$ is strictly increasing and its domain is an open set in $\mathbb{R}$. 

To show \eqref{eq:k_equal}, suppose on the contrary that \eqref{eq:k_equal} is false for some $x_1, \ldots, x_n > 0$, positive integer $k$, and $i \in N \setminus \{1\}$; assume without loss of generality that $i = 2$, which means that
\begin{align*}
f_n((k+1)x_1, kx_2, x_3, \ldots, x_n) \neq f_n(kx_1, (k+1)x_2, x_3, \ldots, x_n).
\end{align*}
Suppose that
\begin{align*}
f_n((k+1)x_1, kx_2, x_3, \ldots, x_n) < f_n(kx_1, (k+1)x_2, x_3, \ldots, x_n);
\end{align*}
the case where the inequality goes in the opposite direction can be handled similarly.
By the continuity of $f_n$, there exists $\epsilon \in (0, x_1)$ such that 
\begin{align} \label{eq:k_eps}
    f_n((k+1)x_1 - \epsilon, kx_2, x_3 \ldots, x_n) < f_n(kx_1 - \epsilon, (k+1)x_2, x_3, \ldots, x_n).
\end{align}
Consider a profile with $m = kn + 1$ goods, where $G' := \{g_1, \ldots, g_{kn} \} = G\setminus\{g_{m}\}$, such that
\begin{itemize}
    \item for each $g \in G'$, we have $u_j(g) = x_j$ for $j \in \{1, 2\}$ and $u_j(g) = x_j/k$ for $j \in N \setminus \{1, 2\}$;
    \item $u_1(g_m) = x_1 - \epsilon$, and $u_j(g_m) = 0$ for $j \in N \setminus \{1\}$.
\end{itemize}
Clearly, this profile admits an allocation where every agent receives positive utility.
Let $A$ be an EF1 allocation chosen by the welfarist rule with function $f_n$ on this profile.
Regardless of whom $g_m$ is allocated to, each agent receives at most $k$ goods from $G'$ in $A$---otherwise, if some agent $j$ receives more than $k$ goods from $G'$, then some other agent receives fewer than $k$ goods from $G'$ by the pigeonhole principle and therefore envies $j$ by more than one good, meaning that $A$ is not EF1. 
Since $|G'| = kn$, every agent receives exactly $k$ goods from $G'$. 
Furthermore, $g_m$ must be allocated to agent $1$; otherwise, the allocation where $g_m$ is allocated to agent $1$ (and all other goods are allocated as in $A$) has a higher welfare than $A$, contradicting the fact that $A$ is chosen by the welfarist rule with function $f_n$. 
The welfare of $A$ must not be smaller than that of another allocation where agent $1$ receives $g_m$ along with $k - 1$ goods from $G'$, agent $2$ receives $k + 1$ goods from $G'$, and every other agent receives $k$ goods from $G'$ each. 
This means that
\begin{align*}
    f_n((k+1)x_1 - \epsilon, kx_2, x_3 \ldots, x_n) \geq f_n(kx_1 - \epsilon, (k+1)x_2, x_3, \ldots, x_n),
\end{align*}
contradicting \eqref{eq:k_eps}. This establishes (i).

It remains to prove (ii). 
Consider any $x_1, \ldots, x_n > 0$ and $y_1, \ldots, y_n \geq 0$ satisfying $\prod_{i=1}^n y_i = 0$.
Let $X := \prod_{i=1}^n x_i > 0$.
Without loss of generality, assume that $y_1 = \cdots = y_k = 0$ and $Y := \prod_{i=k+1}^n y_i > 0$ for some $k \in \{1, \ldots, n\}$ (if $k = n$, the empty product $\prod_{i=k+1}^n y_i$ is taken to be $1$). 
Define $z_1, \ldots, z_n$ by $z_i := (X/2Y)^{1/k}$ for all $i \in \{1, \ldots, k\}$ and $z_i := y_i$ for all $i \in \{ k+1, \ldots, n\}$. Then,
\begin{align*}
    f_n(y_1, \ldots, y_n) &\leq f_n(z_1, \ldots, z_n) \tag*{(since $f_n$ is non-decreasing)} \\
    &= q(z_1 \cdots z_k \cdot z_{k+1} \cdots z_n) \tag*{(by (i) and since all $z_i$'s are positive)}\\
    &= q((X/2Y) \cdot y_{k+1} \cdots y_n) \\
    &= q(X/2) \\
    &< q(X) \tag*{(since $q$ is strictly increasing)} \\
    &= q(x_1 \cdots x_n) \\
    &= f_n(x_1, \ldots, x_n), \tag*{(by (i) and since all $x_i$'s are positive)}
\end{align*}
completing the proof of the theorem.
\end{proof}

\subsection*{Acknowledgments}

This work was partially supported by the Singapore Ministry of Education under grant number MOE-T2EP20221-0001 and by an NUS Start-up Grant.
We thank Pakawut Jiradilok for helpful discussions.

\bibliographystyle{plainnat}
\bibliography{main}

\end{document}